\documentclass{llncs}
\usepackage{makeidx}  % allows for indexgeneration
\usepackage{amsmath,amssymb, color, amscd}
\newcommand{\R}{\mathbb{R}}

\newcommand{\Z}{\mathbb{Z}}

\begin{document}
\title{Symmetries of Quasi-Values}
\author{Ales A. Kubena\inst{1} \and Peter Franek\inst{2}}
\institute{Institute of Information Theory and Automation of the ASCR, \\ Pod Vodarenskou vezi 4, 182 08, Prague, Czech Republic\\
\email{kubena@utia.cas.cz},
\and
Institute of Information Technologies,
Czech Technical University,
Thakurova 9,
Prague 160\,00,
Czech Republic \\
\email{peter.franek@fit.cvut.cz}}

\maketitle              % typeset the title of the contribution

\begin{abstract}
According to Shapley's game-theoretical result, there exists a unique game value of finite cooperative games that satisfies axioms on additivity, efficiency, null-player property and symmetry.
The original setting requires symmetry with respect to arbitrary permutations of players. We analyze the consequences of weakening the symmetry axioms and study quasi-values
that are symmetric with respect to permutations from a~ group $G\leq S_n$. 
We classify all the permutation groups $G$ 
that are large enough to assure a unique $G$-symmetric quasi-value, as well as the structure and dimension of the space of all such quasi-values for a general permutation group $G$.

We show how to construct $G$-symmetric quasi-values algorithmically by averaging certain basic quasi-values (marginal operators).
\end{abstract}
\section{Introduction}
A cooperative game is an assignment of a real number to each subset of a given set of players $\Omega$. This illustrates an economic situation 
where a coalition profit depends on the involved players in a generally non-aditive way.
%(usually, a more general property called super-aditivity is required).  
Several approaches deal with the question of redistributing the generated profit to the individual players in a stable or in a ``fair'' way.
The mathematical theory of cooperative games was developed in forties by Neumann and Morgenstern~\cite{Neumann2007}.
Values of games provide a tool for evaluating the contributions of the individual players such that certain natural axioms are satisfied. 
%A value is a function from cooperative games on a fixed player set $\Omega$ to $\R^\Omega$ satisfying certain natural properties.
The most famous value is the Shapley value introduced in 1953~\cite{Shapley53} that exists and is unique for all finite sets of~players. 
%
%An alternate ways leading to Shapley value are game-value axiomatics. 

There exist many axiomatic systems on game values such that the Shapley value is
their only solution: the original Shapley's axiomatics~\cite{Shapley53}, Neyman's~\cite{Neyman1989}, Young's~\cite{Young85},  van den Brink's~\cite{Brink2001} and Kar's axiomatics~\cite{Kar2002}. One of its important characteristics is the symmetry with respect to any permutation of players. This means, roughly speaking, that the value of a player is calculated only from his contributions to various coalitions and not from his identity. One may consider this to represent the \emph{equity} of players.  
However, this is probably not a realistic assumption in many real-world situations where personal friendships and linkage play a major role. 
Some examples of values with restricted symmetry were studied, such as the \emph{Owen value}~\cite{Owen:1977} or the
\emph{weighted Shapley value} in~\cite{Kalai:1987}, 
and the formal concept of \emph{quasi-value}, where one completely relaxes any symmetry requirement, 
was introduced by Gilboa and Monderer in 1991~\cite{Gilboa1991}. 
It is known that for a particular player set, there exist infinitely many quasi-values.

In this work, we analyze one particular way of weakening the symmetry axiom. We suppose that a group $G$ of permutations of $\Omega$ is given and define  
a~\emph{$G$-symmetric quasi-value} to be any quasi-value symmetric wrt. all permutations in $G$. Informally, the equity of players is restricted to a group of 
permutations of players, not necessarily to all permutations. The group expresses the measure of symmetry. 
If $G$ is the full symmetry group, then the only $G$-symmetric quasi-value is the Shapley value;  
if $G$ is the trivial group, then it carries no symmetry requirement and each quasi-value is $G$-symmetric. 
Our contribution is the classification of all permutation groups $G$ 
of finite sets of players for which there exists a unique $G$-symmetric quasi-value. 
It turns out that while in the infinite setting for non-atomic games one may reduce the group of symmetries in a number of ways~\cite{Monderer:1990,Neyman:2002},
in the finite setting, only few subgroups of the full permutation group assure uniqueness.
Even if the group $G$ acts transitively on $\Omega$ 
(i.e. for any two players $a,b$, there exists a permutation $\pi\in G$ such that $\pi(a)=b$), there may still exist many $G$-symmetric quasi-values different from the Shapley value. 
We also calculate the dimension of the space of all $G$-symmetric quasi-values for a general permutation group $G$.
%We calculate the size of all $G$-symmetric quasi-values, 
%show that it is always an affine space and calculate explicitely its dimension. We also classify all permutation groups $G$ of an $n$-player set such that the $G$-symmetric
%quasi-value is unique.

In the second section, we give the formal definition of $G$-symmetric quasi-value and some necessary definitions from 
group theory, including our original definition of a supertransitive group action. 
In the third section, we show that the space of all $G$-symmetric quasi-values is an affine subspace of the vector space of all values, 
and derive a formula for its dimension.
We further classify all permutation groups $G$ such that there exists a unique $G$-symmetric quasivalue. 
In the fourth section, we give some examples of $G$-symmetric quasi-values and show how more examples can be constructed by averaging the marginal operators.
The last section (Appendix) contains the proof of an auxiliary statement from group theory that we use in the proof of 
Theorem~\ref{thm:uniqueness}. We postpone this technical issue to the end in order to keep the rest of the text fluent.

\section{Definitions and notation}
\subsection{Cooperative games}
%The following definitions are taken from~\cite{Gilles2010}.
Let $\Omega$ be a set of players. In this paper, we always suppose that $\Omega$ is finite.
\begin{definition}
A  \emph{cooperative game} is a function $v: 2^\Omega\to \R$ such that $v(\emptyset)=0$. 
A cooperative game is \emph{additive}, if for all $T,R\in 2^\Omega$, $R\cap T=\emptyset$ implies $v(R\cup T)=v(R)+v(T)$. 
We denote by $\Gamma$ the set of all cooperative games and $\Gamma_1$ the set of all additive cooperative games. 
%
%Cooperative games form a vector space over the real numbers in a natural way.
%If $\Omega$ is a finite set, then there is an isomorphism 
%\begin{equation}
%\label{additive_isom}
%\Gamma_1\simeq \R^\Omega\quad\mathrm{given\,\, by}\quad v\to (v \{i\})_i
%\end{equation}
%where $\{i\}$ is a one-player set for each $i\in\Omega$. In this text, we always suppose that $\Omega$ is finite.
%
A \emph{game value} is an operator $\varphi:\Gamma\to\Gamma_1$. 
For a game value $\varphi$ and $i\in\Omega$, we define $\varphi_i(v):=\varphi(v) (\{i\})$. 
\end{definition} 
For each game $v$, $\varphi(v)$ is uniquelly determined by the numbers $\varphi_i(v)$.
Shapley theorem \cite{Shapley53} proves the existence and uniqueness of a game value $\varphi$ assuming it satisfies the
following four axioms:
\begin{enumerate}
\item\label{linearity} \emph{Linearity}: 
%\footnote{In the original version there is a formally weaker ``additivity'' condition 
%$\varphi(v+v')=\varphi(v)+\varphi(v')$ for all $(\Omega,v), (\Omega,v')\in\Gamma$. 
%Assuming the axiom of choice, there exists other then linear solution of the functional equation $f(x+y)=f(x)+f(y)$, however, they are all eliminated by a further reasonable
%condition on the solution. Linearity is a consequence of the additivity property + continuity, additivity + boundedness on compact sets, additivity + measurability etc.
%All the proofs in this paper are valid independent of the form of this axiom. }
$\varphi(\alpha v+\beta w)=\alpha \varphi(v)+\beta \varphi(w)$ for all $v, w\in\Gamma$ and $\alpha,\beta\in\R$.
\item\label{null-player} \emph{Null-player property}: if $i\in\Omega$ is a ``null-player'' in a game $v$, 
i.e. $v(R\cup\{i\})=v(R)$ for each $R\subseteq\Omega$, then $\varphi_i(v)=0$. 
\item\label{efficiency} \emph{Efficiency}: $\sum_{i} \varphi_i(v)=v(\Omega)$ for all games $v$.
\item\label{symmetry} \emph{Symmetry} (sometimes called \emph{anonymity}): $\varphi(\pi\cdot v)=\pi\cdot \varphi(v)$ for every permutation $\pi$ of $\Omega$, where 
the game $\pi\cdot v$ is defined by 
$
%\label{eq:pi_action}
(\pi\cdot v)(R):=v(\pi^{-1}(R))$
for any $R\subseteq\Omega$.
\end{enumerate}
The value defined by these axioms is called \emph{Shapley value}. 
Axioms \ref{linearity}-\ref{symmetry} are independent. 
Gilles \cite{Gilles2010} and Schmeidler \cite{Schmeidler1969} give examples of values satisfying any 3 of them and not the 4th. 

Any game value satisfying axioms \ref{linearity}, \ref{null-player} and~\ref{efficiency} is called a~\emph{quasi-value}.
In the original economic interpretation, the fourth axiom (Symmetry) is an expression of equity of all the participating players. It can be formulated in a more elegant way
by the commutativity of the following diagram.
\begin{equation}
\begin{CD}
\label{eq:sym}
\Gamma @>\varphi>> \Gamma_1 \\
@VV\pi V   @VV\pi V \\
\Gamma @>\varphi>> \Gamma_1
\end{CD}
\end{equation}
Axiom \ref{symmetry} requires that it commutes for each permutation of players $\pi$. 
%the vertical maps being defined by $(\ref{eq:pi_action})$.

The following definition introduces the main object of our study.
\begin{definition}
\label{def:G-quasi-value}
Let $G$ by a group of permutations of $\Omega$.
A \emph{G-symmetric quasi-value} is a game value that satisfies axioms \ref{linearity}, \ref{null-player}, \ref{efficiency} and such that 
$\varphi(\pi\cdot v)=\pi\cdot \varphi(v)$ for every permutation $\pi\in G$. In other words, diagram $(\ref{eq:sym})$ 
commutes for all $\pi\in G$.
\end{definition}

% EXPLICITNI DEFINICI SHAPLEY-VALUE A HARSANYI DIVIDEND JSEM ZAKOMENTOVAL %%%%%%%%%%%
%The {\bf Shapley value} is a value $\phi$ defined by the formula
%$$
%\phi_i\circ v=\sum_{R\supseteq \{i\}}\frac{\Delta_v(R)}{|R|}
%$$
%where $\Delta_v(R)\in\R$ is a {\bf Harsanyi dividend} of the coalition $R\subseteq \Omega$ defined by
%$$
%\Delta_v(R)=\sum_{T\subseteq R} (-1)^{|R|-|T|} v(T)
%$$
%\end{definition}
%
%\begin{definition}
Throughout this work, we will need the following basis of the space of cooperative games, introduced in Shapley's original paper~\cite{Shapley53}.
\begin{definition}
\label{unanimity}
The \emph{unanimity basis} is the basis $\{u_R\}_{\emptyset\neq R\subseteq \Omega}$ of the vector space of all cooperative games over the set $\Omega$, 
defined by $u_R(S)=1$ if $R\subseteq S$ and $0$ otherwise.
\end{definition}

\subsection{Group theory}
\label{groups}

%Rewise: $S_n$, $A_n$, $S_5\subset S_6$, group action, stabilizer...

We say that a group $G$ \emph{acts} on the set $X$, if $G$ is a subgroup of the group $S_X$ of permutations of $X$. 
Any set $G\cdot x$ is called an \emph{orbit}, or a $G$-orbit of $x$. The set of all $G$-orbits is denoted by $X/G$.
The action of $G$ on $X$ is \emph{transitive}, if for each $x,y\in X$, there exists a $g\in G$ such that $g\cdot x=y$. 
The \emph{stabilizer} of a subset $A\subseteq X$ is the subgroup $G_A$ of all elements $g\in G$ such that $g\cdot A\subseteq A$.
For a subgroup $H$ of $G$, $g\cdot H$ denotes a \emph{left} and $H\cdot g$ a \emph{right coset} of $H$ and any group $H'=g^{-1} H g$ is \emph{conjugate} to $H$.

We introduce here a definition that will help us to describe a property of permutation groups we will need later. 
\begin{definition}
Let $G$ be a group acting on a set $X$. We say that the action is a~\emph{supertransitive action}, if the stabilizer $G_A$ of any subset
$A\subseteq X$ acts transitively on $A$. A permutation group $G\subseteq S_n$ is \emph{supertransitive}, if the stabilizer $G_A$ acts transitively on
each $A\subseteq \{1,\ldots ,n\}$.
\end{definition}

%One important instance of a supertransitive group action is introduced in the following definition.
%\begin{definition}
For any $n$, $S_{n-1}$ may be embedded into $S_n$ as a set of permutations preserving one element. 
However, for $n=6$, there exists an embedding of $S_5$ into $S_6$ different from the standard one. 
This embedding $S_5\hookrightarrow S_6$ may be realized as the action of the projective linear group $PGL(2,5)$ on the projective line over $\Z_5$.
The reader may find the details in the literature~\cite[p. 60-61]{Dixon96}, \cite{Scot07}. We will call this embedding 
an \emph{exotic embedding}. It is well known that such a nonstandard embedding is only one up to conjugation by an element of $S_6$.
In this paper, we only need the property that the image of the exotic embedding is a supertransitive subgroup of $S_6$. This is proved 
in the appendix.
%\end{definition}

\section{Dimension of $G$-symmetric quasi-values}
If a quasi-value is symmetric with respect to a set of permutations, it is also symmetric with respect to any permutation they generate in $S_\Omega$,
hence the set of all symmetries of a quasi-value is always a group.
For a finite set $\Omega$ and a group $G\subseteq S_\Omega$ of permutations,
we denote by  $\mathcal{A}_{G}$ the set of all $G$-symmetric quasi-values. 

We will represent $\mathcal{A}_G$ as a space of matrices.
Each game value $\varphi$ can be represented as a map from $\Gamma$ to $\R^\Omega$ by the natural identification $\Gamma_1\simeq\R^\Omega$.
Choosing the unanimity basis on $\Gamma$ (Def. \ref{unanimity}) and the canonical basis $(e_i)_{i\in\Omega}$ on $\R^\Omega$, 
we may represent linear game values as matrices of the size $|\Omega| \times  (2^{|\Omega|}-1)$.
The null player property applied to the unanimity basis implies $\varphi(u_R)(\{i\})=0$ for each $i\notin R$, because such player $i$ doesn't contribute to any coalition in the game $u_R$.
As a consequence, a matrix $A$ with elements $(a_{iR})_{i\in\Omega,\,\emptyset\neq R\subseteq\Omega}$ corresponds to a linear game value satisfying the null-player-property iff $a_{iR}=0$ for all pairs $(i,R)$
such that $i\notin R$. 
%Validity of this condition follows directly from applying the null-player-property
%on games $u_R$ forming the unanimity basis. 
%
Further, the game value satisfies the efficiency axiom iff for any nonempty $R\subseteq\Omega$, $\varphi(u_R)(\Omega)=1$, which translates to a constraint on matrix coefficients
$\sum_{i\in R}a_{iR}=1$ for each $\emptyset\neq R\subseteq\Omega$. 
The $G$-symmetry of a game value requires $\varphi(g\cdot v)=g\cdot(\varphi(v))$ for any game $v$ and permutation $g\in G$, the action of $G$ on $\Gamma$ defined by $(g\cdot v)(R)=v(g^{-1} R)$.
An element $u_R$ from the unanimity basis satisfies $(g\cdot u_R)(S)=u_R(g^{-1}(S))=u_{gR}(S)$, so the unanimity basis is invariant with respect to
the group action and $g\cdot u_R=u_{gR}$. The symmetry axiom is equivalent to
\begin{equation*}
((g\cdot\varphi)(u_R))(\{i\})=(\varphi(u_{gR}))(\{i\}),
\end{equation*}
for all $i\in\Omega$ and $\emptyset\neq R\subseteq\Omega$.
The left-hand side is equal to $\varphi(u_R)(\{g^{-1}i\})$. So, in~the matrix representation of $\varphi$, the symmetry axiom translates to the condition 
$a_{(g^{-1}i)\,R}=a_{i\,(gR)}$, or simply $a_{iR}=a_{(gi) \,(gR)}$ for all $i\in\Omega$,  $\emptyset\neq R\subseteq\Omega$ and $g\in G$.

Summarizing this, we have the following.
\begin{lemma}
\label{observ}
Choosing the unanimity basis of $\Gamma$ and the canonical basis of $\R^\Omega\simeq\Gamma_1$, $\mathcal{A}_G$ may be identified with a set of matrices $A=(a_{iR})$
with elements satisfying the following equations:
\begin{itemize}
\item{$a_{iR}=0$ if $i\notin R$},
\item{The sum of elements in each column is $1$},
\item{Matrix elements $a_{iR}$ are constant on the orbits of the $G$-action $g\cdot (i,R)=(gi, gR)$.}
\end{itemize}
\end{lemma}
All these conditions are linear equations for matrix elements $a_{iR}$ and they are all satisfied by the Shapley value. 
So, $\mathcal{A}_G$ is a~nonempty affine space.
\begin{theorem}
\label{dimension}
Let $X=\{(i,R);\,i\in R\subseteq\Omega\}$, 
$\chi=\{R;\,\,\emptyset\neq R\subseteq\Omega\}$ and let $G\subseteq S_G$ be a group of permutations acting on sets $X$ and $\chi$,
extending naturally its action on $\Omega$. Then the dimension of $\mathcal{A}_G$ is $|X/G|-|\chi/G|$. Explicitly it can also be expressed as 
\begin{equation}
\label{dimension}
\dim\,\mathcal{A}_G =(\frac{dZ_{G}}{dx_{1}}-Z_{G})|_{(2,2 \ldots 2)}+1
\end{equation}
where $Z_G$ is the cycle index of the group $G$ 
\begin{equation}
\label{cycle_index}
Z_{G}(x_{1}...x_{n})=\frac{1}{|G|}\sum_{\pi\in G}x_{1}^{j_{1}(\pi)}\cdots x_{n}^{j_{n}(\pi)},
\end{equation}
$j_k(\pi)$ denotes the number of cycles of length $k$ in the permutation $\pi$~\cite[p. 85]{analcomb}.
\end{theorem}

\begin{proof}
We will 
%call the $G$-orbits of $X$ ``orbits'' and the $G$-orbits of $\chi$ ``metaorbits'' and
identify elements of $\mathcal{A}_G$ with matrices as described in Lemma~\ref{observ}. 
Let $p: X\to \chi$ be the map $(i,R)\to R$. For any $x=(i,R)\in X$ and $g\in G$, 
$p(gx)=g(p(x))$. %, so an orbit is a refinement of a metaorbit. 
For $\emptyset\neq R\subseteq\Omega$, 
the stabilizer $G_R$ acts on $R$ and $R$ splits into $k_R$ orbits $\{R_1,\ldots, R_{k_R}\}$ 
with respect to this action. 
If $R'=gR$, 
%is on the same metaorbit 
then the stabilizer of $R'$ is $g G_R g^{-1}$ and $g$ maps each $G_R$-orbit $R_i\subseteq R$ 
bijectively onto a $G_{R'}$-orbit $R'_i\subseteq R'$. 
So, $k_R=k_{R'}$ and $|R_i|=|R'_i|$ for $i=1,\ldots, k_R$.
For $m\in\chi/G$, we define $k_m:=k_R$ for any $R\in m$ and $l_{mi}=|R_i|$ for $i=1,\ldots,k_m$.
These numbers are independent on the choice of $R$.

We will say that $m\in\chi/G$ {\it contains} an orbit $Gx\in X/G$, if $p(x)\in m$. 
Each $m\in\chi/G$ contains $k_m$ orbits $\{o_1,\ldots, o_{k_m}\}\subseteq X/G$ and
we may choose real numbers $c_{mi}$ such that $\sum_{i=1}^{k_m} c_{mi} l_{mi}=1$
with $k_m-1$ degrees of freedom. Choosing such numbers $c_{mi}$ for all $m\in\chi/G$
gives 
$$\sum_{m\in \chi/G} (k_m-1)=\sum_{m\in M} k_m - |\chi/G|=|X/G|-|\chi/G|$$
degrees of freedom. Any such choice of $c_{mi}$ defines a matrix of game value 
$$
a_{iR}=\begin{cases}
{c_{mi} \,\,\text{if \,$i\in R_i\subseteq R\in m$}}\\
{0\,\,\text{ if\,\, $i\notin R$}}
\end{cases}
$$
These are exactly matrices $A$ constant on the orbits of $X$ satisfying $\sum_i a_{iR}=1$ for all $R$ and $a_{iR}=0$ for all $i\notin R$.
The number of degrees of freedom for the choice
of $c_{mi}$ is equal to the dimension of $\mathcal{A}_G$. This proves the first part.

Burnside lemma \cite[p. 58]{Rotman95} enables to express the number of orbits of a group action in an explicit way. If a finite group
$H$ acts on a finite set $Y$, then 
\begin{equation}
\label{burnside}
|Y/H|=\frac{1}{|H|} \sum_{h\in H}\,|\{y\in Y\,|\,h(y)=y\}|.
\end{equation}
A permutation $\pi\in G$ fixes those sets $R\subseteq\Omega$ that don't split any cycle of $\pi$. 
There exists $2^{\#\,cycles(\pi)}$ such sets, $2^{\#\,cycles(\pi)}-1$ of them nonempty. So,
$$|\chi/G|=\big(\frac{1}{|G|}\,\sum_{\pi\in G} 2^{\#\,cycles(\pi)}\big)-1.$$
Elements of $X$ fixed by $\pi$ are pairs $(i,R)$ such that $i\in R$, $\pi(i)=i$ and $\pi(R)=R$.
There exists $\#\,fixedpoints(\pi)*2^{\#\,cycles(\pi)-1}$ such pairs. We derived the following equation:
$$
\dim\mathcal{A}_G=\frac{1}{|G|}\big(\sum_{\pi\in G}(\#\mathrm{fixedpoints}(\pi)*2^{\#cycles(\pi)-1})-\sum_{\pi\in G}2^{\#cycles(\pi)}\big)+1.
$$
The statement of the theorem follows from this by a direct computation. $\square$
\end{proof}
The cycle index $Z_G$ is known in a more explicit form than $(\ref{cycle_index})$ for many subgroups of $S_n$ and 
it has also been generalized and computed for finite classical groups~\cite{Fulman:97}.

Further, we will show for which groups $G$ the dimension of $\mathcal{A}_G$ is zero, i.e. for which $G$ the only $G$-symmetric quasi-value 
is the Shapley value.
In Section \ref{groups}, we defined a group $G\subseteq S_\Omega$ to be supertransitive, if
the stabilizer $G_R$ acts transitively on $R$ for each subset $R\subseteq\Omega$. In other words,
if for each $R$ and each $i,j\in R$, there exists a $g\in G$ such that $g(R)=R$ and $g\cdot i=j$. We will show that this condition is equivalent to 
the existence of a unique $G$-symmetric quasi-value.
%For $|\Omega|>2$, the supertransitivity of $G\subset S_\Omega$ is a stronger condition then transitivity.
%For example, the cyclic group $C_3\subseteq S_3$ containing the even permutations $(1,2,3), (2,3,1)$ and $(3,1,2)$ acts transitively on $\{1,2,3\}$ but not supertransitively,
%because no element of $C_3$ fixes $\{1,2\}$ and sends $1$ to $2$. 
%In the following theorem, We will show that $\dim\mathcal{A}_G=0$ is equivalent to the supertransitive action of $G$ on $\Omega$
%and classify all such groups.

\begin{theorem}
\label{thm:uniqueness}
Let $\Omega$ be finite and $G\leq S_\Omega$. There exists a unique $G$-symmetric quasi-value if and only if $G$ acts supertransitively on $\Omega$.
Equivalently, this is if and only one of the following conditions is satisfied:
\begin{itemize}
\item{$G=S_\Omega$, the full symmetric group}
\item{$|\Omega|>3$ and $G=A_\Omega$, the alternating group}
\item{$|\Omega|=6$ and $G$ is the image of an exotic embedding $S_5\hookrightarrow S_6$ (see Section~\ref{groups}).}
\end{itemize}
\end{theorem}
\begin{proof}
We will work with the matrix representation of $\mathcal{A}_G$, described in Lemma~\ref{observ}.
Let $(a_{iR})$ be a matrix representing a value in $\mathcal{A}_G$.

If the action of $G$ on $\Omega$ is supertransitive, then for each 
$\emptyset\neq R \subseteq\Omega$, all elements $\{(i,R);\,i\in R\}$ lie on the same $G$-orbit, so all the corresponding matrix elements 
$a_{iR}$ are equal. The null-player property
implies that $a_{iR}=0$ for $i\notin R$ and together with the efficiency
condition we obtain that for each $i\in R$, $a_{iR}=1/|R|$. This implies uniqueness.

If the action of $G$ on $\Omega$ is not supertransitive, then there exists
a nonempty subset $\tilde{R}\subseteq\Omega$ such that the stabilizer $G_{\tilde{R}}$ 
has not a transitive action on $\tilde{R}$.
So, $\tilde{R}$ contains at least two $G_{\tilde{R}}$-orbits. We may define the matrix $a_{iR}$ as follows.
In the matrix column corresponding to $\tilde{R}$ we choose $a_{i\tilde{R}}=0$ if $i\notin \tilde{R}$ and the other elements $a_{j\tilde{R}}$ arbitrary, constant on $G_{\tilde{R}}$-orbits
and such that $\sum_j a_{j\tilde{R}}=1$. For all $R'$ on the $G$-orbit of $R$, we define the coefficients $a_{iR'}$ in a unique way so that they are constant on the $G$-orbits
and the remaining matrix elements may be equal to elements of the original Shapley matrix.
In this way, we may construct an infinite number
of different $G$-symmetric quasi-values which proves that $\dim \mathcal{A}_G\geq 1$.

For the classification part, it remains to prove that the groups listed in the theorem are exactly the groups acting supertransitively on $\{1,\ldots, n\}$.
The proof of this is technical and we postpone it to the Appendix (Chapter \ref{appendix}). $\square$
\end{proof}

\section{Consequences}
\subsection{Examples}
First we give some examples of groups and $G$-symmetric quasi-values. In all these examples, we assume that the player set $\Omega$ consists of $n$ players.
\bigskip\\
{\bf Example 1.}
Let $G_1=\{\rm{id}\}$ be the trivial group. In this case, any quasi-value is $G_1$-symmetric.  
%The following construction is taken from Gilles~\cite[p. 78]{Gilles2010}.
Consider a selector $\gamma: 2^\Omega \to \Omega$ with $\gamma(R)\in R$ for all $\emptyset\neq R\subseteq\Omega$. Now we define the value $\varphi$ as 
\begin{equation}
\label{example:triv}
\varphi_i(v)=\sum_{i=\gamma(R)} \Delta_v(R)
\end{equation}
where $\Delta_v(R)\in\R$ is a \emph{Harsanyi dividend} of the coalition $R\subseteq \Omega$ defined by
$
\Delta_v(R)=\sum_{T\subseteq R} (-1)^{|R|-|T|} v(T).
$
It was shown in~\cite{Derks:2000} that such values satisfy the axioms for quasi-values.~\footnote{In the matrix representation, such values
correspond to matrices $a_{i\,R}=\delta_{i\gamma(R)}$.}
%This value satisfies efficiency, null player property and linearity. 
The cycle index of the trivial group is $Z(x_1)=x_1^n$ and substituting into $(\ref{dimension})$ yields $\dim\mathcal{A}_{G_1}=n 2^{n-1}-2^n+1$.
However, the number of selectors $\gamma: 2^{\Omega}\to\Omega$ is much larger, so many of the quasi-values defined by $(\ref{example:triv})$ are affine dependent.\footnote{
For $n\geq 4$, $\dim\mathcal{A}_{G_1}$ is strictly smaller than $n!-1$ which implies that the set of marginal operators (defined
in Section~\ref{marginal}) is also affine dependent.}
%An example of such value is any \emph{marginal operator} $m_\pi$ defined in Section~\ref{marginal}, where $\pi$ is a permutation.

\bigskip
{\bf Example 2.} (``Caste system'')
The set $\Omega$ is split into $k$ nonempty disjoint subsets (``castes'') $\Omega_1,\ldots,\Omega_k$ and $G_2$ is chosen so that it guarantees 
equity within 
each $\Omega_i$. Formally, $G_2=\{\pi\in S_\Omega\,|\,\forall i\,\, \pi(\Omega_i)=\Omega_i\}$.

%We start with an example of $G$-symmetric quasi-values. 
%For any game $v$ and caste $r$ we define a new game $v_r$ by $v_r(R)=0$ for $R$ disjoint from $\Omega_r$ and 
%$v_r(R)=v(R\cup (\Omega-\Omega_r))$ otherwise. All players from $\Omega\setminus\Omega_r$ are null-players in $v_r$.
%One can easilly check that the Shapley value $\varphi^r_v\in\mathcal{A}_G$ of $v_r$ is $G$-symmetric.

Some examples of $G_2$-symmetric quasivalues have been described in the literature.
The \emph{Owen value}, defined in~\cite{Owen:1977},
can be obtained as the expected value of \emph{marginal operators} (see Section~\ref{marginal}), 
if we first randomely choose an order of the castes and then the order of the players within each caste.
Another related concept is the \emph{weighted Shapley value}, studied by Kalai and Samet in~\cite{Kalai:1987}.
Here an order of the castes is given and within each
caste, the profit is diveded among players proportional to their \emph{weights}. In the case of equal weights of all players,
the weighted Shapley value is symmetric with respect to all $G_2$-permutations.

The cycle index is $Z_{G_2}=\prod_{r=1}^k\, Z_{S_{\Omega_r}}$. 
We know from the proof of Theorem~\ref{dimension} that 
%the number of metaorbits is 
$|\chi/G|=\frac{1}{|G|} \sum_g 2^{\# cycles(g)}$ for each set $\chi$ with a $G$-action. In particular, for $G=S_n$,
$|2^\Omega/G|=n+1$, because $S_n$-orbits of $2^\Omega$ are $O_s=\{R\subseteq\Omega\,|\,|R|=s\}$ for $s=0,1,\ldots, n$.
This enables as to calculate
$$
Z_{S_n}|_{(2,\ldots,2)}=\frac{1}{n!}\sum_{\pi\in S_n} 2^{j_1(\pi)+\ldots+j_n(\pi)}=
\frac{1}{n!}\sum_{\pi\in S_n} 2^{\# cycles(\pi)}=|2^\Omega/S_n|=n+1.
$$
If $G=S_n$, then the Shapley value is the only game value, so it follows from Theorem \ref{dimension}
that $(\frac{dZ_{S_n}}{dx_1}-Z_{S_n})|_{(2,\ldots,2)}+1=0$ and $\frac{dZ_{S_n}}{dx_1}|_{(2,\ldots,2)}=n$.
So, for $G_2=\prod_{r=1}^k S_{\Omega_r}$
$$
\frac{dZ_{G_2}}{dx_{1}}|_{(2,2...2)}=\big(\sum_{r=1}^{k}\frac{dZ_{S_{\Omega_{r}}}}{dx_{1}}
\prod_{s\neq r}Z_{S_{\Omega_{s}}}\big)|_{(2,2...2)}=\sum_{r=1}^{k}|\Omega_{r}|\prod_{s\neq r}(1+|\Omega_{s}|)
$$
and 
$$
\dim\mathcal{A}_{G_2}=(\sum_{r=1}^{k}\frac{|\Omega_{r}|}{1+|\Omega_{r}|}-1)\prod_{r=1}^{k}(1+|\Omega_{r}|)+1.
$$
For the case of two castes $k=2$ this simplifies to $|\Omega_1|\times |\Omega_2|$.
\bigskip \\
{\bf Example 3.} (Cyclic group)
This example illustrates that transitive group action does not imply a unique $G$-symmetric quasi-value.
If $G_3$ is the cyclic group $C_n\subseteq S_n$, the cycle index is 
$Z_{C_n}=\frac{1}{n}\sum_{f|n}\phi(f) x_f^{n/f}$, where $\phi(f)$ is the Euler totient function 
$\phi(f)=p_1^{k_1-1}(p_1-1)\ldots p_r^{k_r-1}(p_r-1)$, where $f=p_1^{k_1}\ldots p_r^{k_r}$ is
the prime number decomposition.\cite[p. 86]{analcomb}. Substituting into the formula in Theorem \ref{dimension}
gives
$$
\dim{\mathcal{A}_{G_3}}=2^{n-1}-\frac{1}{n}\sum_{f|n} \phi(f) 2^{n/f} +1.
$$
In the case of $n=3$, the dimension turns out to be $2^2-\frac{1}{3}(2^3+2\times 2)+1=1$, so there exists a one-dimensional space of 
quasi-values symmetric with respect to cyclic permutations of players.
\subsection{Shapley-value as an expected value of non-uniformly distributed marginal vectors}
\label{marginal}
Suppose that $\Omega=\{1,2,\ldots,n\}$, i.e. an order is given on the set of player.  For a game $v\in\Gamma$ and a permutation $\pi\in S_n$,
we may define a quasi-value $m_\pi$ by $(m_\pi)(v)_{\pi(1)}=v(\pi(1))$ and 
$$(m_\pi(v))_{\pi(i)}=v(\{\pi(1),\pi(2),\ldots, \pi(i)\})-v(\{\pi(1),\pi(2),\ldots, \pi(i-1)\})$$
for $i=2,\ldots, n$. We call $m_\pi$ the \emph{marginal operator} and $m_\pi(v)$ the \emph{marginal vector}~\cite[p. 19]{Branzei2005}. 
It corresponds to a situation where the players arrive in the order $\pi(1), \pi(2),\ldots, \pi(n)$ and each player is assigned the value of
his or her contribution to the coalition of all players that have arrived before.
The evaluation of $m_\pi$ on a game $u_R$ from the unanimity basis is $m_\pi(u_R)(\{\pi(i)\})=u_R(\pi(1),\ldots,\pi(i))-u_R(\pi(1),\ldots, \pi(i-1))$
which is equal to $1$ if and only if $\pi(i)\in R$ and $\pi(j)\notin R$ for $j>i$ and $0$ otherwise. 
After the identification \ref{observ}, we can represent $m_\pi$ is as a matrix
$$
({m}_{\pi})_{iR}=\begin{cases}
1\,\, \textrm{iff}\,\, i\in R\,\, \textrm{and $\pi^{-1}(i)=\max \pi^{-1}(R)$}\\
0\,\,\textrm{ otherwise}.
\end{cases}
$$
A theorem of Weber \cite{Weber1988} shows that if $\pi$ is a random permutation taken from a uniform distribution on $S_n$ then for any game
$v$, the expected value of a marginal operator $m_\pi$ is the Shapley value. 
%Similarly, we may construct a quasi-value from any probability distribution $(A^\pi)_{\pi\in S_n}$ on the set of permutations. 
This can be generalized to the following statement.
\begin{proposition}
Let $G$ be a subgroup of $S_n$ and $A^\pi$ be a probability distributioin on $S_n$ constant on the right cosets $\{G\cdot \pi\}_\pi$, 
i.e. $A^\pi=A^{g \pi}$ for all $g\in G$ and $\pi\in S_n$. Then $\sum A^\pi\,m_\pi$ is a $G$-symmetric quasi-value. 
\end{proposition}
\begin{proof}
We will show that the identity holds if evaluated on games from the unanimity basis of $\Gamma$. For the game $u_R$ (Definition \ref{unanimity}), we start with the following
equation:
\begin{equation}
\label{marginality-invariance}
(g \cdot m_\pi) (u_R)=m_{g\, \pi} (u_{gR}).
\end{equation}
To prove this, we evaluate both sides on $\{i\}$ and rewrite the left-hand side to the equivalent equation
$$
(m_\pi (u_{R})) (\{g^{-1}(i)\})=(m_{g\,\pi} (u_{gR})) (\{i\}).
$$
Both sides are equal to $1$ if and only if $\pi^{-1}(g^{-1}(i))=\max \pi^{-1}(R)$ and $0$ otherwise, which proves $(\ref{marginality-invariance})$ for all $R\subseteq\Omega$, $i\in\Omega$
and $g\in G$.
The $G$-symmetry of $\sum_{\pi\in S_n} A^\pi m_\pi$ follows from 
\begin{align*}
& \big(g\cdot \sum_{\pi\in S_n} A^\pi m_\pi\big) (u_R)=\sum_{\pi\in S_n} A^{\pi} (g\cdot m_\pi)(u_R)=\sum_{\pi\in S_n} A^\pi m_{g \pi} (u_{gR})=\\
& =\sum_{\pi\in S_n} A^{g \pi} m_{g\,\pi} (u_{gR})=\sum_{g\,\pi=\pi'\in S_n} A^{\pi'} m_{\pi'} (g\cdot u_R)=\big((\sum_{\pi'\in S_n} A^{\pi'} m_{\pi'}) \cdot g\big) (u_R)
\end{align*}
where we used $(\ref{marginality-invariance})$ in the second and $A^{\pi}=A^{g \pi}$ in the third equality. $\square$
\end{proof}

An immediate consequence of the classification Theorem $\ref{thm:uniqueness}$ is that for $|\Omega|>3$ any quasi-value symmetric 
with respect to the alternating group $A_n$
is already the Shapley value. It follows from the last proposition that 
$\sum_\pi A^\pi m_\pi$ is the Shapley value not only for $A^\pi=\frac{1}{n!}$ but also for $A^\pi=\frac{s}{n!}$ for $\pi$ even 
and $A^\pi=\frac{2-s}{n!}$ for $\pi$ odd, $s\in [0,2]$. In fact, there are many other possibilities how to express the Shapley value 
as a convex combination of marginal operators.
The space of all quasi-values on $\Omega$ is  $(n2^{n-1}-2^{n}+1)$-dimensional and the set of all probability distributions
on $S_n$ is a $(n!-1)$-dimensional convex region in $\R^{n!}$, so there are {at least} $n!-n2^{n}+2^{n-1}-2$ degrees of freedom 
for the choice of a distribution $A^\pi$ such that $\sum_\pi A^\pi m_\pi=\textrm{Shapley}$. 

Exponentially many (with respect to~$n$) of these probability distributions $A^\pi$ can be constructed as follows. Choose $\Omega_{0}\subseteq\Omega,\;|\Omega_{0}|>3$
and define $S_{0}$ to be a group of all permutations $\pi$ acting identically on $\Omega\setminus\Omega_{0}$.
Choose $\alpha\in (0,2)$ and define a probability distribution on $S_n$ by
\[
A^{\pi}(\Omega_0)=
\begin{cases}
\frac{1}{n!}\textrm{ if }\pi\notin S_{0}\\
\frac{\alpha}{n!}\textrm{ if }\pi\in S_{0}\textrm{ and }\pi\textrm{ is even }\\
\frac{2-\alpha}{n!}\textrm{ if }\pi\in S_{0}\textrm{ and }\pi\textrm{ is odd }
\end{cases}
\]
One can verify that the corresponding expected value of marginal operators $m_\pi$ is the Shapley value. For a set
$\{\Omega_{1},\Omega_{2}, \ldots, \Omega_{k}\}$ s.t. $\Omega_{i}\nsubseteq\Omega_{j}$ for all $i$ and $j$, 
the vectors $(A^\pi(\Omega_i)-\frac{1}{n!})_i\in\R^{n!}$ are linearly independent and the distributions 
$(A^\pi(\Omega_i))_i$ are affine independent.
\section{Appendix}
\label{appendix}
Here we finish the proof of Theorem $\ref{thm:uniqueness}$ by the classification of supertransitive groups. 
Our proof is based on a classification of set-transitive permutation groups given by Beamont and Peterson in 1955~\cite{Beaumont1955}. 
Another proof of the supertransitive groups classification was given by Michal Jordan on mathoverflow~\cite{Jordan2011}.
\begin{theorem}
$G$ is a supertransitive subgroup of $S_n$ if and only if one of the following conditions holds:
\begin{itemize}
\item $G$ is the full symmetric group $S_n$ for some $n$,
\item $G$ is the alternating group $A_n$ for $n>3$,
\item $G$ is conjugate to the image of an exotic embedding of $S_5$ to $S_6$.
\end{itemize}
\end{theorem}
\begin{proof}
Let $G\subseteq S_n$ be a group of permutations acting supertransitively on $\{1,\ldots,n\}$. This means that the stabilizer of each $A\subseteq\{1,\ldots, n\}$ acts
transitively on $A$. Let $B\subseteq\{1,\ldots,n\}$ and $i,j\notin B$. Then $G$ acts transitively on $B\cup\{i,j\}$ and there exists a permutation 
$\pi\in G$ taking $B\cup\{i\}$ to $B\cup\{j\}$ such that $\pi(i)=j$. This implies that for each $A$ and $B$ s.t. $|A|=|B|>1$, there exists a permutation $\pi\in G$
s.t. $\pi(A)=B$. If $|A|=|B|=1$, the same is true because supertransitivity implies transitivity. We have shown that if the action of $G$ is supertransitive, it is also set-transitive.

If $G$ has a supertransitive action on $\{1,\ldots,n\}$, then its order has to be divisible by each $k\leq n$, because each $k$-element set $A$ is isomorphic to $G/G_A$, hence
$|G|=|A|\times |G_A|$. So, $G$ has to be divisible by the least common multiple of $\{1,\ldots, n\}$.

Beamont and Petrson classified all set-transitive permutation groups in \cite{Beaumont1955}. It follows that such subgroups of $S_n$ are exactly the full symmetric group $S_n$ for any $n$, the alternating group $A_n$ for $n>2$ and $5$ exceptions. The first and second exceptions are subgroups of $S_5$ of order $10$, resp. $20$. These groups cannot have a supertransitive action on $\{1,\ldots,5\}$, because the lowest common multiple of $\{1,\ldots,5\}$ is $60$. 
Two other exceptions in Beamont's classification are subgroups of $S_9$ of orders 504 and 1512. These numbers are not divisible by the lowest common multiple of $\{1,\ldots, 9\}$
so we can exclude them as well. The last exception is a subgroup of $S_6$ of order 120. This groups is equivalent to the exotic embedding of $S_5$ to $S_6$ and we will show that
it acts supertransitively on $S_6$.

In \cite{Janusz1982}, the authors realize this group action on $\{1,\ldots, 6\} $ as the conjugate action of $S_5$ on its six Sylow $5$-subgroups. Using this realisation, we may show that this action is supertransitive by direct calculation. Let as denote the Sylow $5$-subgroups by $I=\langle (12345)\rangle$, $II=\langle (12354)\rangle$, $III=\langle(12435)\rangle$, $IV=\langle (12453)\rangle$, $V=\langle (12534))\rangle$ 
and $VI=\langle (12543)\rangle$. 
An elementary calculation shows that the image of a transposition in $S_5$ is the product of three disjoint transpositions in $S_6$, 
e.g. $(1,2)\in S_5\,\mapsto\, (I,VI)\,(II, IV) (III,V)$ in the above realisation.  
Together with the set-transitivity of this $S_5$-action, this implies $2$-supertransitivity.
The image of a $3$-cycle in $S_5$ is a product of two disjoint $3$-cycles in $S_6$, which implies $3$-supertransitivity.
Similarly, the image of a $4$-, resp. $5$-cycle in $S_5$ is a $4$-, resp. $5$-cycle in $S_6$, which implies $4$- and $5$-supertransitivity.

It remains to prove that $A_n$ is supertransitive if and only if $n>3$.
First note that $A_2=\{id\}$, reps. $A_3=\langle (123)\rangle$ are not supertransitive, 
because no element of these groups takes 1 to 2 and preserves $\{1,2\}$. 
Let $n>3$ and $A\subseteq \{1,\ldots, n\}$ be a $k$-set. 
If $k<n-1$, then any permutation of $A$ can be extended to an even permutation of $\{1,\ldots,n\}$. 
If $k=n-1>2$, then for each $i,j\in A$, there exists an even permutation of $A$ taking $i$ to $j$.
This can be extended to an even permutation of $\{1,\ldots, n\}$, acting identically on the complement of $A$.
$\square$\end{proof}

\bigskip
\section{Acknowledgements}
%\bigskip

We would like to thank to Michal Jordan for his mathematical remarks and discussion on mathoverflow.
This work was supported by M{\v S}MT project number OC10048  and by the institutional research plan AV0Z100300504 and by the Excelence project P402/12/G097 DYME – Dynamic Models in Economics of GA\v CR.

\bibliographystyle{abbrv}
\bibliography{shapley_symmetries}

\end{document}